\documentclass[a4paper,UKenglish,dvipsnames]{lipics-v2016-modified}
 
\usepackage{microtype}

\usepackage[dvipsnames]{xcolor}
\usepackage[fleqn]{amsmath}
\usepackage{amssymb,latexsym,wasysym,dsfont,stmaryrd,bm}
\usepackage{relsize}
\usepackage{xspace}
\usepackage{graphicx}
\usepackage{gastex}
\usepackage{tikz}
\usetikzlibrary{arrows,automata,shapes.geometric,decorations.pathmorphing,backgrounds,positioning,fit,petri,calc}
\usepackage{wrapfig}
\usepackage{macros}
\usepackage{manfnt}              
\usepackage{enumitem}
\setlist[enumerate]{itemsep=0mm}
\setlist[itemize]{itemsep=0mm,topsep=1pt,partopsep=0pt}

\newif\ifdraft\drafttrue

\ifdraft
\newcommand\modedraft[1]{#1}
\newcommand\todo[1]{{\color{purple}[\textbf{To do:} #1]}}
\newcommand\bmcomment[1]{{\footnotesize \color{blue}[#1 - \textbf{Bastien}]}}
\newcommand\amcomment[1]{{\footnotesize
    \color{OliveGreen}[#1 - \textbf{Nello}]}}
\newcommand\rbcomment[1]{{\footnotesize \color{Orange}[#1 - \textbf{Raphael}]}}

\else
\newcommand\modedraft[1]{}
\newcommand\todo[1]{}
\newcommand\bmcomment[1]{}
\newcommand\amcomment[1]{}
\newcommand\rbcomment[1]{}

\fi

\title{Decidability results for ATL with imperfect
  information and perfect recall\footnote{This project
    has received funding from the European Union's Horizon 2020
    research and innovation programme under the Marie Sklodowska-Curie
    grant agreement No 709188.}}

\author[1]{Rapha\"el Berthon}
\author[2]{Bastien Maubert}
\author[3]{Aniello Murano}
\affil[1]{\'Ecole Normale Supérieure de Rennes, France\\
  \texttt{raphael.berthon@ens-rennes.fr}}
\affil[2]{Universit\`a degli Studi di Napoli Federico II, Italy\\
  \texttt{bastien.maubert@gmail.com}}
\affil[3]{Universit\`a degli Studi di Napoli Federico II, Italy\\
  \texttt{murano@na.infn.it}}

\authorrunning{R. Berthon, B. Maubert and A. Murano} 

\Copyright{Rapha\"el Berthon, Bastien Maubert and Aniello Murano}


\theoremstyle{plain}
\newtheorem{proposition}[theorem]{Proposition}

\newcommand\UElogo{%
\begin{tikzpicture}[remember picture,overlay]
\node[anchor=south,yshift=4.2cm,xshift=2cm] at (current page.south) {\includegraphics[height=2.5em]{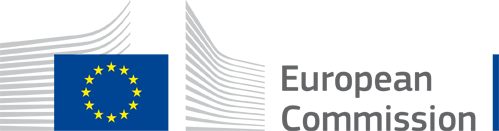}};
\end{tikzpicture}%
}

\begin{document}


\maketitle
\UElogo

\begin{abstract}
Alternating-time Temporal Logic (\ATLs) is a central logic for
multiagent  systems. It has been extended in
various ways, notably with imperfect
information (\ATLsi). Since the model-checking problem
against \ATLsi for agents  with perfect recall is undecidable, studies have mostly
focused either on
agents without memory, or on alternative semantics to
retrieve decidability. In this work,  we establish new, strong
decidability results for agents with perfect recall. We first prove a
meta-theorem that allows the transfer of decidability results for classes
of multiplayer games with imperfect information, such as games with
hierarchical observation, to the model-checking
problem for \ATLsi. We also establish that
model checking  \ATLs with strategy context and imperfect
information for \emph{hierarchical instances} is decidable.
\end{abstract}






\section{Introduction}
\label{sec-intro}

In formal system  verification, \emph{model checking} is a
well-established method to automatically check the correctness of a system~\cite{CE81,QS81,CGP02}.  It consists in  modelling the
system as a  mathematical
structure, expressing its
desired behaviour as a formula from some suitable logic, and  checking whether the model
satisfies the formula. 
In the nineties, interest has arisen in the verification of
\emph{multiagent systems} (MAS),
 in which various entities (the \emph{agents})
 interact and can form coalitions to attain certain objectives. This
 led to the development of logics that allow reasoning about  strategic abilities
 in MAS~\cite{AHK02,Pau02,JH04,HWW07,AGJ07,BJ14}.  

 \emph{Alternating-time
  Temporal Logic} (\ATLs), introduced by Alur, Henzinger,
and Kupferman~\cite{AHK02}, plays a central role  in this line of work. This logic, interpreted on \emph{concurrent
  game structures},   extends \CTLs with \emph{strategic
  modalities}. These modalities
   allow one to reason about the existence of strategies
for coalitions of agents to force the system's behaviour to satisfy certain
temporal properties. \ATLs has been extended in many ways,
and among these extensions an important one is \emph{\ATLs with
  strategy context}~\cite{DBLP:conf/lfcs/BrihayeLLM09,DBLP:journals/iandc/LaroussinieM15}. 
In \ATLs,  strategies of all agents are  forgotten at each new strategic
modality. In \ATLs with strategy context (\ATLssc) instead they are
stored in a  \emph{strategy context}, and are forgotten only
when replaced by a new strategy or when the formula explicitly
unbinds the agent from her strategy. Thanks to this additional
expressive power, \ATLssc can express important game theoretic concepts
such as the existence of Nash Equilibria~\cite{DBLP:journals/iandc/LaroussinieM15}. 

In many real-life scenarios, such as poker, 
agents do not always know precisely what is the current state of the
system. Instead, they have a partial view, or observation, of the
state. This fundamental feature of MAS is called \emph{imperfect
  information}, and it is known to quickly bring about undecidability
when involved in strategic problems, especially when agents have
\emph{perfect recall} of the past, which is a usual and important
assumption in games with imperfect information and epistemic temporal
logics~\cite{fagin1995reasoning}. For instance solving multiplayer
games with imperfect information and perfect recall, \ie, deciding the
existence of a distributed winning strategy in such games, is already
undecidable for reachability objective, as proven by Peterson, Reif
and Azhar~\cite{peterson2001lower}.  Since such games are easily
captured by \ATLs with imperfect information (\ATLsi), model checking
\ATLsi with perfect recall is also undecidable~\cite{AHK02}.

However it is known that
restricting attention to cases where
some sort of hierarchy exists on the 
different agents' information yields decidability for several problems related to
the existence of strategies.
Synthesis of distributed systems, which implicitly uses perfect
recall and is
  undecidable in
general~\cite{DBLP:conf/focs/PnueliR90}, is decidable for hierarchical
architectures~\cite{DBLP:conf/lics/KupfermanV01}. Actually, for
branching-time specifications, distributed synthesis is
decidable exactly on architectures free from \emph{information forks},
for which the problem can be reduced to the hierarchical
case~\cite{DBLP:conf/lics/FinkbeinerS05}. For richer
specifications from alternating-time
logics, being free of information forks is no longer sufficient, but  distributed synthesis is decidable precisely  on
hierarchical architectures~\cite{DBLP:conf/atva/ScheweF07}. 
Similarly, solving
 multiplayer games with imperfect information and  perfect recall,
 \ie, checking for the existence of winning distributed strategies,
 is decidable for $\omega$-regular winning conditions when  there is a hierarchy among players, each one observing more
than those below~\cite{peterson2002decision,DBLP:conf/lics/KupfermanV01}.
Recently, it has been proven that this    assumption can  be relaxed
while maintaining decidability:
  the
problem remains decidable if the hierarchy can change along a play, or even
if transient phases without such a hierarchy are
allowed~\cite{DBLP:conf/atva/BerwangerMB15}. 




{\textbf{Our contribution.}} In this work we establish several
decidability results for
model checking    \ATLsi with perfect recall, with and
without strategy context, all related to notions of hierarchy. Our
first result is a  theorem that
allows the transfer of decidability results for classes of multiplayer
games with imperfect information,
such as those mentioned above, to
the model-checking problem for \ATLsi. This theorem essentially states
that if solving multiplayer games with imperfect information, perfect
recall and omega-regular objectives is decidable on some class of
concurrent game structures, then 
model checking \ATLsi with perfect recall is also decidable on this
class of models (a simple bottom-up algorithm that evaluates
innermost strategic modalities in every state of the model
suffices). As a direct consequence we easily obtain new decidability results
for the model checking of \ATLsi on several classes of concurrent game structures.

Our second contribution considers \ATLs with imperfect information and strategy context
(\ATLssci).
Because there are in general infinitely many possible strategy
contexts, the bottom-up approach used for \ATLsi cannot be used
here. Instead we build upon the proof presented
in~\cite{DBLP:journals/iandc/LaroussinieM15} to establish the
decidability of model checking \ATLssc, by reduction to the
model-checking problem for Quantified \CTLs (\QCTLs). The latter 
extends \CTLs with second-order quantification on atomic
 propositions, and it has been studied in a number of
 works~\cite{Sis83,Kup95,KMTV00,french2001decidability,DBLP:journals/corr/LaroussinieM14}.
\QCTLsi, an imperfect-information extension of \QCTLs, has  recently been
 introduced, and its model-checking problem was proven decidable for the
 class of \emph{hierarchical formulas}~\cite{BMMRV17}.
In this paper we define a notion of  \emph{hierarchical instances} for
the \ATLssci model-checking problem: informally, an \ATLssci formula $\phi$
together with a concurrent game structure $\CGS$ is a hierarchical instance if
outermost strategic modalities in $\phi$ concern agents who
observe less in $\CGS$. We adapt the proof
from~\cite{DBLP:journals/iandc/LaroussinieM15} and reduce the
model-checking problem for \ATLssci on hierarchical instances to the
model-checking problem for hierarchical \QCTLsi formulas. We obtain that model checking
hierarchical instances of
\ATLssci with perfect recall is decidable. 
 
{\textbf{Related work.}} The model-checking problem for
\ATLsi is known to be decidable
when agents have no memory~\cite{Sch04}, and the case of agents with bounded memory
reduces to that of no memory. Another way to retrieve decidability is
to assume that all agents in a coalition have the same information,
either because their observations of the system are the same, or
because they can communicate and share their
observations~\cite{DBLP:conf/dalt/GuelevD08,DBLP:journals/corr/abs-1006-1414,DBLP:journals/jancl/GuelevDE11,DBLP:conf/prima/KazmierczakAJ14}. This
idea was also used recently to establish a decidability result for
\ATLssci~\cite{DBLP:journals/corr/LaroussinieMS15} when
all agents have the same observation of the game.

The results we establish here thus strictly extend previously known
results on the decidability of model checking \ATLsi and \ATLssci with
perfect recall and standard semantics, and
they hold for vast, natural classes of instances, that all rely on
notions of hierarchy, which seems to be inherent to all decidable
cases of strategic problems for multiple entities with imperfect
information and perfect recall.

{\textbf{Outline.}}  After setting some basic definitions in
Section~\ref{sec-prelim}, we present  our transfer theorem and its
various corollaries concerning   the model
checking problem for \ATLsi in Section~\ref{sec-ATL}. In Section~\ref{sec-ATLsc} we 
 prove that when restricted to hierarchical instances, model
checking \ATLssci is decidable, and we conclude in Section~\ref{sec-conclusion}. 


\section{Preleminaries}
\label{sec-prelim}

Let $\Sigma$ be an alphabet. A \emph{finite} (resp. \emph{infinite}) \emph{word} over $\Sigma$ is an element
of $\Sigma^{*}$ (resp. $\Sigma^{\omega}$). The empty word is
 noted $\epsilon$, and
$\Sigma^{+}=\Sigma^{*}\setminus\{\epsilon\}$. The \emph{length} of a
  word is $|w|\egdef 0$ if $w$ is the empty word $\epsilon$, if $w=w_{0}w_{1}\ldots
w_{n}$ is a finite nonempty word then $|w|\egdef n+1$, and for an infinite word $w$ we let $|w|\egdef
\omega$. Given a word $w$ and $0\leq i,j\leq |w|-1$, we let $w_{i}$ be the
letter at position $i$ in $w$ and $w[i,j]$ be the subword of $w$ that starts
at position $i$ and ends at position $j$.
For $n\in\setn$ we let $[n]\egdef\{1,\ldots,n\}$. 
Finally,  let us fix a countably
infinite set of \emph{atomic propositions} $\AP$
and let $\APf\subset\AP$ be some finite subset of atomic propositions.

\subsection{Kripke structures}
\label{sec-kripke}

  A \emph{Kripke structure} over $\APf$ is a tuple
  $\KS=(\setstates,\relation,\lab)$ where
  $\setstates$
  is a set
  of \emph{states},
  $\relation\subseteq\setstates\times\setstates$ is a
  left-total\footnote{\ie, for all $\state\in\setstates$, there exists
    $\state'$ such that $(\state,\state')\in\relation$.}
  \emph{transition relation} and $\lab:\setstates\to 2^{\APf}$ is a
  \emph{\labeling function}. 

A \emph{pointed Kripke structure} is a
  pair $(\KS,\state)$ where $\state\in\KS$. 
A \emph{path} in a structure   $\KS=(\setstates,\relation,\lab)$  is
an infinite word $\spath$ over $\setstates$
such that for all $i\in\setn$,
$(\spath_{i},\spath_{i+1})\in \relation$. For 
$\state\in\setstates$,  $\Paths(\state)$ is the set of all
paths that start in $\state$.


\subsection{Infinite trees}
\label{sec-trees}
Let $\Dirtree$ be a finite set. An \emph{$\Dirtree$-tree} $\tree$ 
 is a
 nonempty set of words $\tree\subseteq \Dirtree^+$ such that
\begin{itemize}
  \item\label{p-root} there exists $\racine\in\Dirtree$,  called the
    \emph{root} of $\tree$, such that each
    $\noeud\in\tree$ starts with $\racine$; 
  \item if $\noeud\cdot\dir\in\tree$ and $\noeud\neq\epsilon$, then
    $\noeud\in\tree$, and
  \item if $\noeud\in\tree$ then there exists $\dir\in\Dirtree$ such that $\noeud\cdot\dir\in\tree$.
\end{itemize}

The elements of a tree $\tree$ are called \emph{nodes}.  
  If 
 $\noeud\cdot\dir \in \tree$, we say that $\noeud\cdot\dir$ is a \emph{child} of
 $\noeud$. 
Similarly to Kripke structures, a \emph{\tpath} is an infinite sequence of nodes $\tpath=\noeud_0\noeud_1\ldots$
such that for all $i$, $\noeud_{i+1}$ is a child of
$\noeud_i$,
and $\tPaths(\noeud)$ is the set of \tpaths
 that start in node $\noeud$. 
An \emph{$\APf$-\labeled $\Dirtree$-tree}, or
\emph{$(\APf,\Dirtree)$-tree} for short, is a pair
$\ltree=(\tree,\lab)$, where $\tree$ is an $\Dirtree$-tree called the
\emph{domain} of $\ltree$ and
$\lab:\tree \rightarrow 2^{\APf}$ is a \emph{\labeling}.

\begin{definition}[Tree unfoldings]
  \label{sec-unfoldings}
  Let $\CKS=(\setstates,\relation,\lab)$ be a Kripke structure over $\APf$, and let $\state\in\setstates$. 
  The \emph{tree-unfolding of $\CKS$ from $\state$} is the 
  $(\APf,\setstates)$-tree $\unfold{\state}=(\tree,\lab')$, where
    $\tree$ is the set
    of all finite  paths that start in $\state$, and
    for every $\noeud\in\tree$,
    $\lab'(\noeud)=\lab(\last(\noeud))$.
\end{definition}


\section{\ATLstitle with imperfect information}
\label{sec-ATL}

In this section we recall the syntax and semantics of \ATLs with
imperfect information and synchronous perfect-recall semantics, or \ATLsi for short, and establish a
meta-theorem on the decidability of its model-checking problem.

\subsection{Definitions}
\label{sec-def}

We first introduce the models of the logics we study.  For the rest of the
paper, let us fix  a non-empty finite
set of \emph{agents} $\Ag$ and a non-empty finite set of
\emph{moves} $\Mov$. 

\begin{definition}
  \label{def-CGS}
  A \emph{concurrent game structure with imperfect information} (or
  \CGS for short) over
  $\APf$  is a
  tuple $\CGS=(\setpos,\trans,\val,\{\obseq\}_{\ag\in\Ag})$ where
  $\setpos$ is a non-empty finite set of \emph{positions},
  $\trans:\setpos\times \Mov^{\Ag}\to \setpos$ is a
  \emph{transition function}, $\val:\setpos\to 2^{\APf}$ is a \emph{labelling function}
  and for each agent $\ag\in\Ag$,
  $\obseq\;\subseteq \setpos\times\setpos$ is an equivalence
  relation. 
\end{definition}


In a position $\pos\in\setpos$, each agent $\ag$ chooses a move $\mova\in\Mov$, 
and the game proceeds to position
$\trans(\pos, \jmov)$, where $\jmov\in \Mov^{\Ag}$ stands for the \emph{joint move}
$(\mova)_{\ag\in\Ag}$ (note that we assume 
$\trans(\pos,\jmov)$ to be defined for all $\pos$ and
$\jmov$\footnote{\label{footnote-models}This assumption, as well as the choice of a unique
  set of moves for all agents, is made to ease
  presentation. All the results presented here also hold when the set
  of available moves depends on the agent and the position.}).
For each position $\pos\in\setpos$, $\val(\pos)$ is the finite set of atomic
propositions that hold in $\pos$, and for $\ag\in\Ag$, equivalence
relation  $\obseq$
represents the observation of agent $\ag$: for two positions
$\pos,\pos'\in\setpos$, $\pos\obseq \pos'$ means that agent~$\ag$ cannot tell the
difference between $\pos$ and $\pos'$. We may  write $\pos\in\CGS$ for $\pos\in\setpos$.
A \emph{pointed \CGS} $(\CGS,\pos)$ is a \CGS $\CGS$
together with a position $\pos\in\CGS$.

In Section~\ref{sec-modelcheck-ATL} we also use
\emph{nondeterministic \CGS}, which are as in Definition~\ref{def-CGS}
except that they have a \emph{transition relation}
$\trans\subseteq\setpos\times\Mov^{\Ag}\times\setpos$ instead of a
transition function. In a position $\pos$,
after every agent  has chosen a move, forming a joint move
$\jmov\in\Mov^{\Ag}$, a special
agent called Nature (not in $\Ag$) chooses a next position
$\pos'$ such that $(\pos,\jmov,\pos')\in\trans$
(see~\cite{DBLP:conf/atva/BerwangerMB15} for detail). In the
following, unless explicitly
specified, \CGS always  refers to deterministic \CGS.
The following definitions also concern deterministic \CGS, but they
can be adapted to nondeterministic ones in an obvious way.

A \emph{finite} (resp. \emph{infinite}) \emph{play} is a finite (resp. infinite)
word $\fplay=\pos_{0}\ldots \pos_{n}$ (resp. $\iplay=\pos_{0} \pos_{1}\ldots$)
such that for all $i$ with $0\leq i<|\fplay|-1$ (resp. $i\geq 0$), there exists a joint move $\jmov$
such that $\trans(\pos_{i}, \jmov)=\pos_{i+1}$. A finite (resp. infinite)
play $\fplay$ (resp. $\iplay$) \emph{starts} in
a position $\pos$ if $\fplay_{0}=\pos$ (resp. $\iplay_{0}=\pos$).
We let $\FPlay(\CGS,\pos)$ be the set of plays, either finite or infinite,  that start
in $\pos$.

In this work we consider
agents with synchronous perfect recall, meaning that the observational
equivalence relation for each agent~$\ag$ is extended to finite plays
the following way: $\fplay \sim_{\ag} \fplay'$ if $|\fplay|=|\fplay|'$
and $\fplay_{i}\sim_{\ag}\fplay'_{i}$ for every $i\in\{0,\ldots, |\fplay|-1\}$. A
\emph{strategy for agent~\ag} is a function
$\strat:\setpos^{+}\to\Mov$ such that $\strat(\fplay)=\strat(\fplay')$ whenever $\fplay \sim_{\ag}
\fplay'$. The latter  constraint captures the essence of
imperfect information, which is that agents can base their strategic
choices only on the information available to them, and removing this constraint yields the semantics of classic
\ATL with perfect information.

A \emph{strategy profile} for a coalition $\coal\subseteq
\Ag$ is a mapping $\strat_{\coal}$ that assigns a strategy to each agent $\ag\in
\coal$; for $\ag\in\coal$, we may write
$\strat_{a}$ instead of $\strat_{\coal}(a)$. An infinite play $\iplay$ \emph{follows} a
strategy profile $\strat_{\coal}$ for a coalition $\coal$ if for all $i\geq
0$, there exists a joint move $\jmov$ such that
$\trans(\iplay_{i},\jmov)=\iplay_{i+1}$ and for each $a\in \coal$,
$\mov_{a}=\strat_{a}(\iplay[0,i])$. For a strategy profile
$\strat_{\coal}$ and a position $\pos\in\setpos$, we define the
outcome $\out(\pos,\strat_{\coal})$ of
$\strat_{\coal}$ in $\pos$ as the set of
infinite plays that start in $\pos$ and follow $\strat_{\coal}$.

The syntax of \ATLsi is the same as that of \ATLs, and is given by
the following grammar:
\[  \phi ::= p \mid \neg \phi \mid \phi \ou \phi \mid \Estrat \phi
  \mid \X\phi \mid \phi\until\phi,\]
  where $p\in\AP$ and $\coal\subseteq \Ag$.

$\X$ and $\until$ are the classic \emph{next}
and \emph{until} operators, respectively, while the \emph{strategic} operator
$\Estrat$ quantifies on strategy profiles for coalition $\coal$.
  
The semantics of \ATLsi  is defined with regards to a \CGS $\CGS=(\setpos,\trans,\val,\{\obseq\}_{\ag\in\Ag})$,
an infinite play $\iplay$ and a position $i\geq 0$ along this play, by induction on formulas:
\[\begin{array}{ll}
  \CGS,\iplay,i\models  p & \mbox{if }p\in\val(\iplay_{i})  \\
  \CGS,\iplay,i\models  \neg \phi & \mbox{if } \CGS,\iplay,i\not\models \phi  \\
  \CGS,\iplay,i\models  \phi\ou\phi' & \mbox{if } \CGS,\iplay,i\models \phi \mbox{ or
  }\CGS,\iplay,i\models \phi' \\
  \CGS,\iplay,i\models  \Estrat\phi & \mbox{if there exists a strategy
    profile $\strat_{\coal}$ s.t.}\\ & \mbox{for all
  $\iplay'\in\out(\iplay_{i},\strat_{\coal})$,
  $\CGS,\iplay',0\models\phi$}\\
\CGS,\iplay,i \models \X\phi & \mbox{if }\CGS,\iplay,i+1\models\phi \\
\CGS,\iplay,i \models \phi\until\phi' & \mbox{if there exists $j\geq i$
  s.t. }\CGS,\iplay,j\models\phi' \mbox{ and,}\\
& \mbox{for all $k$ s.t. $i\leq k
  <j$, } \CGS,\iplay,k\models\phi.
\end{array}\]

An \ATLsi formula $\phi$ is \emph{closed} if every temporal
operator ($\X$ or $\until$) in $\phi$ is in the scope of a strategic
operator $\Estrat$. Since the semantics of a closed formula $\phi$
does not depend on the future, we may write $\CGS,\pos\models \phi$
if $\CGS,\iplay,0\models\phi$ for any infinite play $\iplay$
that starts in $\pos$.

The \emph{model-checking problem for \ATLsi}
consists in deciding, given a closed \ATLsi formula $\phi$ and a
finite pointed \CGS $(\CGS,\pos)$, whether $\CGS,\pos\models\phi$.

\subsection{Model checking \ATLsititle}
\label{sec-modelcheck-ATL}

It is well known that the model-checking problem for \ATLsi is
undecidable for agents with perfect recall~\cite{AHK02}, as it can easily express the
  existence of distributed winning strategies
 for multiplayer reachability games with imperfect information and
 perfect recall, which was proved
 undecidable by Peterson, Reif and Azhar~\cite{peterson2001lower}. A
 direct proof of this undecidability result for \ATLsi is also presented
 in~\cite{DBLP:journals/corr/abs-1102-4225}. However, there are
 classes of multiplayer games with imperfect information that are
 decidable. For many years, the only known decidable case was that of
 hierarchical games, in which there is a total preorder among players,
 each player observing at least as much as those below her in this
 preorder \cite{peterson2002decision,DBLP:conf/lics/KupfermanV01}. Recently, this result has been extended by relaxing the
 assumption of hierarchical observation. In particular, it has been
 shown that the problem remains decidable if the hierarchy
 can change along a play, or if transient phases without such a hierarchy
 are allowed~\cite{DBLP:conf/atva/BerwangerMB15}. We
 establish that these results transfer to the model-checking problem
 for \ATLsi.

We remind that a concurrent game with imperfect information is a pair 
$((\CGS,\pos),W)$ where $(\CGS,\pos)$ is a pointed \emph{nondeterministic} \CGS and $W$ is 
 a property of infinite plays called the \emph{winning condition}. 
 The
\emph{strategy problem} is, given such a game, to decide whether there exists a strategy profile
for the grand coalition $\Ag$  to enforce the
winning condition against Nature (for more details  see, \eg, \cite{DBLP:conf/atva/BerwangerMB15}). 

Before stating our transfer theorem we need to introduce a couple of
additional notions. First we introduce a notion of abstraction over 
a group of agents. Informally, abstracting a \CGS $\CGS$ over an agent
consists in erasing her from
the group of agents  and letting Nature play for her in $\CGS$.

\begin{definition}[Abstraction]
  \label{def-abstraction}
Let $\coal\subseteq\Ag$ be a group of agents and
let   $\CGS=(\setpos,\trans,\lab,\{\obseq\}_{\ag\in\Ag})$ be a
  \CGS. The
  \emph{abstraction of $\CGS$ from $\coal$} is the nondeterministic
  \CGS over set
  of agents $\Ag\setminus\coal$ defined as $\abs{\coal}\egdef(\setpos,\trans',\lab,\{\obseq\}_{\ag\in\Ag\setminus\coal})$, where
  for every $\pos\in\setpos$ and $\jmov\in\Mov^{\Ag\setminus\coal}$,
  \[(\pos,\jmov,\pos')\in\trans'\mbox{ if }\;\exists
  \jmov'\in\Mov^{\coal} \sthat \trans(\pos,(\jmov,\jmov'))=\pos'.\]
\end{definition}

Thanks to this notion we can define the following problem: 

\begin{definition}[$\coal$-strategy problem]
  \label{def-A-strategy}
The $\coal$-strategy problem takes as input a pointed  \CGS
$(\CGS,\pos)$, a set $\coal\subseteq\Ag$ of agents and a winning
condition $W$, and returns the answer to the strategy problem for the game $((\abs{\Ag\setminus\coal},\pos),W)$.
\end{definition}
 The $\coal$-strategy problem for $(\CGS,\pos)$ with
winning condition $W$ thus consists in deciding whether there is a
 strategy profile for agents in $\coal$ to enforce $W$ against
 everybody else.

 Finally we introduce the following notion, which simply captures the
 change of initial position in a game from a position $\pos$ to another
 position $\pos'$ reachable from $\pos$:

 \begin{definition}[Initial shifting]
   \label{def-shifting}
   Let $\CGS$ be a  \CGS and let $\pos,\pos'\in\CGS$. The pointed \CGS
   $(\CGS,\pos')$ is an 
 \emph{initial
   shifting} of $(\CGS,\pos)$
if $\pos'$ is reachable from $\pos$ in $\CGS$.
 \end{definition}

 We are now ready to state our first result.
\begin{theorem}
  \label{theo-meta}
  If $\class$ is a class of pointed \CGSs closed under initial
  shifting and such that
 the $\coal$-strategy 
  problem with $\omega$-regular objective is decidable on $\class$, then model
  checking \ATLsi is decidable on $\class$.
\end{theorem}

\begin{proof}
  Let $\class$ be such a class of pointed \CGSs, and let
  $(\phi,(\CGS,\pos))$ be an instance of the model-checking problem
  for \ATLsi on $\class$. A bottom-up algorithm consists in evaluating
  each innermost subformula of $\phi$ of the form $\Estrat \phi'$,
  where $\phi'$ is thus an \LTL formula, on each position $\pos'$ of
  $\CGS$ reachable from $\pos$. Evaluating $\Estrat \phi'$ on $\pos'$
  amounts to solving an instance of the $\coal$-strategy
  problem\footnote{Observe that if $\coal=\Ag$ then
    $\abs{\Ag\setminus\coal}=\CGS$, and Nature thus does not do
    anything. This is coherent with the fact that for agents with
    perfect recall
    $\Estrat[\Ag]\phi\equiv \E\phi$, where $\E$ is the \CTL path
    quantifier, even for imperfect information.}
  with $\omega$-regular objective (recall that \LTL properties are
  $\omega$-regular).  By assumption $(\CGS,\pos)\in\class$, and
  because $\class$ is closed by initial shifting and $\pos'$ is
  reachable from $\pos$, we have that $(\CGS,\pos')\in\class$.  Also
  by assumption, the $\coal$-strategy problem for $\omega$-regular
  winning conditions is decidable on $\class$. We thus have an
  algorithm to evaluate each $\Estrat \phi'$ on each $\pos'$.  One can
  then mark positions of the game with fresh atomic propositions
  indicating where these formulas hold, and repeat the procedure until
  all strategic operators have been eliminated. It then remains to
  evaluate a boolean formula in the initial position $\pos$.
\end{proof}

Let us recall for which classes of nondeterministic \CGSs the strategy
problem is known to be decidable. A (nondeterministic or
deterministic) \CGS $\CGS$ 
 has
\emph{hierarchical observation} if there exists a total preorder $\pref$
over $\Ag$ such that if $\ag\pref\agb$ and $\pos\obseq\pos'$, then $\pos\obseq[\agb]\pos'$.
This notion was refined in \cite{DBLP:conf/atva/BerwangerMB15} to take
into account the agents' memory, using the notion of \emph{information
  set}: for a finite play $\fplay\in\FPlay(\CGS,\pos)$ and an agent $\ag$, the
\emph{information set} of agent $\ag$ after $\fplay$ is
$\infset(\fplay)\egdef\{\fplay'\in\FPlay(\CGS,\pos)\mid\fplay \obseq\fplay'\}$.
A finite play $\fplay$  yields \emph{hierarchical
  information} if there is a total preorder $\pref$ over $\Ag$ such that
if $\ag\pref\agb$, then $\infset(\fplay)\subseteq\infset[\agb](\fplay)$.
If all finite plays in $\FPlay(\CGS,\pos)$ yield hierarchical
information for the same preorder over agents, $(\CGS,\pos)$ yields
\emph{static hierarchical information}. If this preorder can vary
depending on the play, $(\CGS,\pos)$ yields \emph{dynamic hierarchical
information}. The last generalisation consists in allowing for
transient phases without hierarchy: if every infinite
play in $\FPlay(\CGS,\pos)$ has infinitely many prefixes that yield
hierarchical information,  $(\CGS,\pos)$ yields \emph{recurring
  hierarchical information}.

\begin{proposition}
  \label{prop-closed-abstraction}
Hierarchical observation as well as  static, dynamic and recurring
hierarchical information are preserved by abstraction.
\end{proposition}

\begin{proof}
  Abstraction removes agents without affecting observations of
  remaining ones. The result thus follows from the
  respective definitions of hierarchical observation and of static, dynamic and recurring hierarchical information.
\end{proof}

\begin{proposition}
  \label{prop-closed-shifting}
Hierarchical observation as well as  static, dynamic and recurring
hierarchical information are 
preserved by initial shifting.
\end{proposition}

This is obvious for hierarchical observation. For the other cases
  we establish Lemma~\ref{lem-hierarchy}
below. It is then easy to check that
Proposition~\ref{prop-closed-shifting} holds. 
\begin{lemma}
  \label{lem-hierarchy}
  If a finite play $\pos\cdot\fplay\cdot\pos'\cdot\fplay'$ yields hierarchical
  information in $(\CGS,\pos)$, so does $\pos'\cdot\fplay'$ in $(\CGS,\pos')$,
   with the same preorder among agents.
\end{lemma}

\begin{proof}
  Assume that $\pos\cdot\fplay\cdot\pos'\cdot\fplay'$ yields
  hierarchical information in $(\CGS,\pos)$ with preorder $\pref$ over
  $\Ag$. Suppose towards a contradiction that there are agents
  $\ag,\agb\in\Ag$ such that $\ag\pref\agb$ but
  $\infset(\pos'\cdot\fplay')\not\subseteq\infset[\agb](\pos'\cdot\fplay')$. This
  means that there is $\pos'\cdot\fplay''\in\FPlay(\CGS,\pos')$ such
  that $\pos'\cdot\fplay'\obseq \pos'\cdot\fplay''$ but
  $\pos'\cdot\fplay'\not\obseq[\agb] \pos'\cdot\fplay''$. By
  definition of synchronous perfect recall relations we then have that
  $\pos\cdot\fplay\cdot\pos'\cdot\fplay'\obseq
  \pos\cdot\fplay\cdot\pos'\cdot\fplay''$ and
  $\pos\cdot\fplay\cdot\pos'\cdot\fplay'\not\obseq[\agb]
  \pos\cdot\fplay\cdot\pos'\cdot\fplay''$. This implies that
  $\infset(\pos\cdot\fplay\cdot\pos'\cdot\fplay')\not\subseteq
  \infset[\agb](\pos\cdot\fplay\cdot\pos'\cdot\fplay')$, which
  contradicts the fact that $\ag\pref\agb$. Therefore for all agents
  $\ag,\agb$ such that $\ag\pref\agb$ we have
  $\infset(\pos'\cdot\fplay')\subseteq\infset[\agb](\pos'\cdot\fplay')$,
  and thus $\pos'\cdot\fplay'$ yields hierarchical information with
  preorder $\pref$.
\end{proof}


Let $\classobs$ (resp. $\classstat$, $\classdyn$, $\classrec$) be the class of
pointed \CGS with hierarchical observation  (resp. static, dynamic, recurring hierarchical
information). We instantiate Theorem~\ref{theo-meta} to obtain three
decidability results for \ATLsi.

\begin{theorem}
  \label{theo-decidability-obs}
Model checking  \ATLsi is decidable on the class of \CGSs with
hierarchical observation.
\end{theorem}

\begin{proof}
  By Proposition~\ref{prop-closed-shifting}, $\classobs$ is closed
  under initial shifting. It is proven in~\cite{DBLP:conf/lics/KupfermanV01} that
 the strategy problem is decidable for games with hierarchical observation
 and $\omega$-regular objectives.
Since, by Proposition~\ref{prop-closed-abstraction}, all
pointed nondeterministic \CGS obtained by abstracting agents from \CGS in
$\classobs$ also yield  hierarchical
observation, we get that the $\coal$-strategy problem with
$\omega$-regular objectives is decidable on $\classobs$.
We can therefore apply Theorem~\ref{theo-meta} on $\classobs$.
\end{proof}

It is proven in
  \cite{DBLP:conf/atva/BerwangerMB15} that
 the strategy problem with $\omega$-regular objectives is also decidable for games with static
 hierarchical information and for games with dynamic
hierarchical information. Since
Proposition~\ref{prop-closed-abstraction} and
Proposition~\ref{prop-closed-shifting} also hold for $\classstat$ and $\classdyn$,
with the same argument as in the proof of
Theorem~\ref{theo-decidability-obs}, we obtain the following results
as consequences of Theorem~\ref{theo-meta}:

\begin{theorem}
  \label{theo-decidability-stat}
Model checking  \ATLsi is decidable on the class of \CGSs with static hierarchical information.
\end{theorem}

\begin{theorem}
  \label{theo-decidability-dyn}
Model checking  \ATLsi is decidable on the class of \CGSs with dynamic hierarchical information.
\end{theorem}

Note that in fact, since $\classobs\subset\classstat\subset\classdyn$,
Theorem~\ref{theo-decidability-obs} and
Theorem~\ref{theo-decidability-stat} are also obtained as
corollaries of Theorem~\ref{theo-decidability-dyn}, but we wanted to
illustrate how Theorem~\ref{theo-meta} can be applied to obtain
decidability results for different classes of \CGS.


\begin{remark}
The last result in \cite{DBLP:conf/atva/BerwangerMB15} establishes
that the strategy problem is decidable for games with recurring
hierarchical information, but only for \emph{observable}
$\omega$-regular winning conditions, \ie, when all agents can tell
whether a play is winning or not. Now considering \ATLsi on
$\classdyn$ we could require atomic propositions to be observable for
all agents; in that case we could evaluate the inner-most strategy
quantifiers using the above-mentioned result. But then the fresh
atomic propositions that mark positions where these subformulas hold
(see the proof of Theorem~\ref{theo-meta}) would not, in general, be
observable by all agents. So on 
$\classrec$
we could obtain a decision procedure for the
fragment of \ATLsi without nested non-trivial strategy quantifiers,
where non-trivial means for coalitions other than the empty coalition
or the one made of all agents (which, we recall, are simply the \CTL
path quantifiers). We do not state it explicitly 
because it does not seem of much interest.
\end{remark}

Concerning complexity,
the strategy problem for games with imperfect information and
hierarchical observation is already  nonelementary \cite{DBLP:conf/focs/PnueliR90,peterson2001lower}, hence
the following result:

\begin{corollary}
  \label{cor-complexity}
  Model checking  \ATLsi is nonelementary on
  games with hierarchical observation, hence also for games with
  static or dynamic hierarchical information.
\end{corollary}



We now turn to \ATL with imperfect information and strategy
context, and study its model-checking problem.


\section{\ATLititle with strategy context}
\label{sec-ATLsc}

While in \ATL strategies for all agents are forgotten each time a new
strategy quantifier is met, in \ATL with strategy context (\ATLsc)
\cite{DBLP:conf/lfcs/BrihayeLLM09,DLM10,DBLP:journals/iandc/LaroussinieM15} agents keep using the
same strategy as long as the formula does not say otherwise.
In this section we consider \ATLsc
with imperfect information (\ATLsci). As far as we know, the
only existing work on this logic is
\cite{DBLP:journals/corr/LaroussinieMS15}, which proved
 its model-checking problem to be decidable in the case where
all agents have the same observation of the game. We extend
significantly this
result by establishing that the model-checking problem
is decidable as long as strategy quantification is
\emph{hierarchical}, in the sense that if there is a strategy
quantification for agent $a$ nested in a strategy quantification for
agent $b$, then $b$ should observe no more than $a$. In other terms,
innermost strategic quantifications should concern agents who observe more.

\subsection{Syntax and semantics}
\label{sec-syntax-ATLsci}

The models are still \CGSs.  
To remember which agents are currently bound to a strategy, and what
these strategies are, the semantics uses
\emph{strategy contexts}.
Formally, a   strategy context for a set of agents $\coala\subseteq\Ag$  is a strategy profile
 $\scon$. We define the composition of strategy contexts as
   follows. If $\scon[\coala]$ is a strategy context for $\coala$ and
   $\scon[\coal]$ is a new strategy profile for coalition  $\coal$,
we let   $\scon[\coal]\comp\scon$ be the strategy context for
$\coal\union\coala$ defined as 
   $\scon[\coal\union \coala]:a \mapsto
   \begin{cases}
     \scon[\coal](a) & \mbox{if }a\in \coal, \\
     \scon(a) & \mbox{otherwise}
   \end{cases}
   $.
   
   So if $a$  is assigned a 
   strategy by $\strat_{\coal}$, her strategy in
   $\scon[\coal]\comp\scon$ is $\strat_{\coal}(a)$. If
   she is not assigned a strategy by $\strat_{\coal}$ her strategy
   remains the one given by $\scon$, if any.
   
   Also, given a strategy context $\scon$ and a set of agents
   $\coal\subseteq\Ag$, we let $\forget{\scon}$ be the strategy context obtained by restricting $\scon$ to the domain
   $\coala\setminus \coal$.

   Finally, because agents who do not change their strategy keep
   playing the one they were assigned, if any, we cannot
    forget the past at each strategy quantifier, as 
    in the  semantics of \ATLsi (see Section~\ref{sec-def}). We thus define the outcome of
   a strategy profile $\strat_{\coal}$ after a finite play $\fplay$,
   written $\out(\fplay,\strat_{\coal})$, as the set of infinite plays
   $\iplay$ that start with $\fplay$ and then follow $\strat_{\coal}$:
   $\iplay\in\out(\fplay,\strat_{\coal})$ if $\iplay=\fplay\cdot\iplay'$
 for some $\iplay'$, and 
 for all $i\geq
|\fplay|-1$, there exists a joint move $\jmov\in\Mov^{\Ag}$ such that
$\trans(\iplay_{i}, \jmov)=\iplay_{i+1}$ and for each $a\in \coal$,
$\mov_{a}=\strat_{a}(\iplay[0,i])$.
   
   To differentiate from  \ATLs, in \ATLssc the strategy quantifier
   for a coalition $\coal$ is written $\Estratsc$ instead of $\Estrat$.  \ATLssc also has an
   additional operator, $\relstrat$, that releases agents in
   $\coal$ from their current strategy, if they have one. 
The syntax of \ATLssci is the same as that of \ATLssc and  is thus
given by the following grammar:
\[  \phi ::= p \mid \neg \phi \mid \phi \ou \phi \mid \Estratsc \phi \mid \relstrat \phi
  \mid \X\phi \mid \phi\until\phi,\]
  where $p\in\AP$ and $\coal\subseteq \Ag$.

  \begin{remark}
    \label{rem-syntax}
    In \cite{DBLP:journals/iandc/LaroussinieM15} the syntax of \ATLssc
     contains in addition operators $\Estratsc[\overline{\coal}]$ and $\relstrat[\overline{\coal}]$ for
    complement coalitions. While they add expressivity when the set of
    agents is not fixed, and are thus of interest when considering 
    expressivity or satisfiability, they are
    redundant if we consider model
    checking, which is our case in this work. To simplify presentation we thus
    choose not to consider them here. 
  \end{remark}
  
  The semantics of \ATLssci is defined with
  regards to a \CGS
  $\CGS=(\setpos,\trans,\val,\{\obseq\}_{\ag\in\Ag})$, an infinite
  play $\iplay$, a position $i\in\setn$ along this play, and a strategy
  context $\scon$. The  semantics is defined by induction
  on formulas:
\[\begin{array}{ll}
  \CGS,\iplay,i\modelssc  p & \mbox{if }p\in\val(\iplay_{i})  \\
  \CGS,\iplay,i\modelssc  \neg \phi & \mbox{if } \CGS,\iplay,i\not\modelssc \phi  \\
  \CGS,\iplay,i\modelssc  \phi\ou\phi' & \mbox{if } \CGS,\iplay,i\modelssc \phi \mbox{ or
  }\CGS,\iplay,i\modelssc \phi' \\
  \CGS,\iplay,i\modelssc  \Estratsc\phi & \mbox{if there exists a strategy
    profile $\strat_{\coal}$ s.t.}\\ & \mbox{for all
  $\iplay'\in\out(\iplay[0,i],\strat_{\coal}\comp\scon)$, }
  \CGS,\iplay',i\modelssc[\strat_{\coal}\comp\scon]\phi\\
\CGS,\iplay,i \modelssc \relstrat\phi & \mbox{if }\CGS,\iplay,i\modelssc[\forget{\scon}]\phi \\
\CGS,\iplay,i \modelssc \X\phi & \mbox{if }\CGS,\iplay,i+1\modelssc\phi \\
\CGS,\iplay,i \modelssc \phi\until\phi' & \mbox{if there exists $j\geq i$
  s.t. }\CGS,\iplay,j\modelssc\phi'\\
& \mbox{and, for all $k$ such that $i\leq k
  <j$, } \CGS,\iplay,k\modelssc\phi.
\end{array}\]

The notion of closed formula is as defined in Section~\ref{sec-def}
and once more,
the semantics of a closed formula $\phi$
being independent from the future, we may write $\CGS,\pos\modelssc \phi$
instead of $\CGS,\iplay,0\modelssc\phi$ for any infinite play $\iplay$
that starts in position $\pos$. We also write $\CGS,\pos\models\phi$ if
$\CGS,\pos\modelssc[\strat_{\emptyset}]\phi$, that is if $\phi$ holds in
$\pos$ with the empty strategy context.

The \emph{model-checking problem for \ATLssci}
consists in deciding, given a closed \ATLssci formula $\phi$ and a
finite pointed \CGS $(\CGS,\pos)$, whether $\CGS,\pos\models\phi$.

We now present  \QCTLs with imperfect information, or \QCTLsi for short, before
proving our main result on the model-checking problem for \ATLssci by
reducing it to the model-checking problem for a decidable fragment of \QCTLsi.


\subsection{\QCTLstitle with imperfect information}
\label{sec-QCTL}

Quantified \CTLs, or \QCTLs for short, is an extension of \CTLs with
second-order quantifiers on atomic propositions that has been
well studied~\cite{Sis83,Kup95,KMTV00,DBLP:journals/corr/LaroussinieM14}.
It has recently been  further extended to take into account
imperfect information, resulting in the logic called \QCTLs with
imperfect information, or \QCTLsi~\cite{berthon2016qctli,BMMRV17}. We briefly
present this logic, as well as a  decidability result on its
model-checking problem proved in \cite{berthon2016qctli,BMMRV17} and that
we rely on
to establish our result on the model checking of \ATLssci.

Imperfect information is incorporated  into \QCTLs by
considering Kripke models with  internal structure in the form of
local states,  like in
distributed
  systems (see for instance~\cite{halpern1989complexity}), and then
 parameterising
   quantifiers on atomic propositions   with observations that define what
  portions of the  states a quantifier can
  ``observe''. The semantics is then adapted to capture the idea
  of quantification on atomic propositions being made with  partial
  observation.

Let us fix a collection
$\{\setlstates_{i}\}_{i\in [n]}$ of $n$ disjoint finite sets
of \emph{local states}. We also let $\Dirtreei[n]=\setlstates_{1}\times\ldots\times\setlstates_{n}$.

\begin{definition}
A \emph{compound Kripke structure} (\CKS) over $\APf$ is a Kripke structure
  $\CKS=(\setstates,\relation,\lab)$ such that 
  $\setstates\subseteq \Dirtreei[n]$. 
\end{definition}



The syntax of \QCTLsi is  that of \QCTLs, except that
quantifiers over atomic propositions are parameterised by a set of
indices that defines what local states the quantifier can
``observe''.
  It is thus defined by the following grammar:
  \begin{align*}
  \phi\egdef &\; p \mid \neg \phi \mid \phi\ou \phi \mid \E \phi \mid
  \existsp[p]{\obs} \phi  \mid \X \phi \mid
  \phi \until \phi
\end{align*}
where $p\in\AP$ 
and $\obs\subset \setn$
is a finite set of indices.
We use standard abbreviations:
$\top\egdef p\ou\neg p$, $\perp\egdef\neg\top$, $\F\phi \egdef \top \until \phi$, $\G\phi \egdef
\neg \F \neg \phi$ and $\A\phi \egdef \neg\E\neg\phi$.

A finite set  $\obs\subset\setn$ is called an \emph{observation}, and
two states  $\state=(\lstate_{1},\ldots,\lstate_{n})$ and
$\state'=(\lstate'_{1},\ldots,\lstate'_{n})$ 
are \emph{$\obs$-indistinguishable},
written $\state\oequiv\state'$, if for all $i\in [n]\inter\obs$, it
holds that
$\lstate_{i}=\lstate'_{i}$.

The intuition is that a quantifier with observation $\obs$ must choose the
valuation of atomic propositions \emph{uniformly} with respect to
$\obs$.   
Note that in~\cite{berthon2016qctli}, two semantics are considered for
\QCTLsi, just like in~\cite{DBLP:journals/corr/LaroussinieM14} for
\QCTLs: the structure semantics and the tree semantics. In the former,
 formulas are evaluated directly on the
structure, while in the latter the structure is first unfolded into an
infinite tree.
Here we only present
the tree semantics, as it is this one that allows us to capture agents with
perfect recall. 
But we first need a few more definitions.

  For $p\in\AP$, two labelled trees $\ltree=(\tree,\lab)$ and
  $\ltree'=(\tree',\lab')$ are \emph{equivalent modulo $p$},
  written $\ltree\Pequiv\ltree'$, if $\tree=\tree'$ and for each node $\noeud\in\tree$,
  $\lab(\noeud)\setminus\{p\}=\lab'(\noeud)\setminus\{p\}$. So
  $\ltree\Pequiv\ltree'$ if they are the same trees, except for the
  labelling of proposition $p$.

This notion of equivalence modulo $p$ is the one used to define
quantification on atomic propositions in  \QCTLs: intuitively, an
existential quantification over $p$ chooses a new labelling for
valuation $p$, all else remaining the same, and the evaluation of the formula continues from the
current node with the new labelling. For imperfect
information we need to express the fact that this new labelling for a
proposition is done uniformly with regards to the quantifier's observation. 

First, we define the notion of indistinguishability between two nodes
in the unfolding of a \CKS.
Let  $\obs$ be an observation, let $\tree$ be an $\Dirtreei[n]$-tree
(which may be obtained by unfolding some pointed \CKS), and let
   $\noeud=\state_{0}\ldots\state_{i}$ and
   $\noeud'=\state'_{0}\ldots\state'_{j}$ be two nodes in $\tree$. 
   The nodes
   $\noeud$ and $\noeud'$ are
   \emph{$\obs$-indistinguishable}, written $\noeud\oequivt\noeud'$, if
   $i=j$ and for all $k\in \{0,\ldots,i\}$, we have
   $\state_{k}\oequiv\state'_{k}$. Observe that this definition
   corresponds to the notion of synchronous perfect recall in \CGS
   (see Section~\ref{sec-def}). 
We  now define what it means for the labelling of an atomic
proposition to be uniform with regards to an observation.
   
   \begin{definition}
     \label{def-uniform}
          Let $\ltree=(\tree,\lab)$ be a labelled $\Dirtreei[n]$-tree,
          let $p\in\AP$ be an atomic proposition  and
          $\obs\subset\setn$ an observation. 
Tree $\ltree$ is \emph{$\obs$-uniform in $p$} if for every pair of nodes
 $\noeud,\noeud'\in\tree$ such that $\noeud\oequivt\noeud'$, we have
 $p\in\lab(\noeud)$ iff $p\in\lab(\noeud')$.
   \end{definition}

 The satisfaction relation $\modelst$ ($t$ is for \emph{tree semantics}) is
now  defined  as follows, where   $\ltree=(\tree,\lab)$ is
a \labeled $\Dirtreei[n]$-tree, 
 $\tpath$ is a path in $\tree$ and $i\in\setn$ a position along that branch:
\[\begin{array}{ll}
  \ltree,\tpath,i\modelst p & \mbox{if }  p\in\lab(\tpath_{i}) \\
    \ltree,\tpath,i\modelst \neg \phi &
    \mbox{if } \ltree,\tpath,i\not\modelst \phi \\
        \ltree,\tpath,i\modelst  \phi \ou \phi'&
    \mbox{if } \ltree,\tpath,i \modelst \phi \mbox{ or
    }\ltree,\tpath,i\modelst \phi' \\
    \ltree,\tpath,i\modelst \E\phi &
    \mbox{if } \mbox{there exists }\tpath'\in\tPaths(\tpath_{i})\mbox{ such
      that }\ltree,\tpath',0\modelst \phi \\
  \ltree,\tpath,i\modelst \existsp{\obs} \phi &
  \mbox{if there exists }\ltree'\Pequiv[p]\ltree \mbox{ such that
  }\ltree'\mbox{ is $\obs$-uniform in $p$ and }\ltree',\tpath,i\modelst\phi\\
    \ltree,\tpath,i\modelst \X\phi &
    \mbox{if }  \ltree,\tpath,i+1\modelst \phi \\
        \ltree,\tpath,i\modelst \phi\until\phi' &
    \mbox{if } \mbox{there exists }j\geq i \mbox{ such that
    }\ltree,\tpath,j\modelst\phi' \mbox{ and for }i\leq k <j,\; \ltree,\tpath,j\modelst\phi
\end{array}\]

Similarly to \ATLssci, we say that a \QCTLsi formula is
\emph{closed} if all temporal operators are in the scope of a path
quantifier. The semantics of such formulas depending only on the
current node, for a closed formula $\phi$ we 
 may write  $\ltree\modelst\phi$ for $\ltree,\racine\modelst\phi$,
where $\racine$ is the root of $\ltree$, and given a pointed \CKS $(\CKS,\state)$  and a
\QCTLsi formula $\phi$, we write $\CKS,\state\modelst\phi$ if $\unfold[\CKS]{\state}\modelst\phi$.

\begin{remark}
  \label{rem-syntax-QCTLi}
  In~\cite{berthon2016qctli} the syntax is presented  with
  path formulas distinguished from state formulas, and the semantics
  is defined accordingly. To make the presentation more uniform with
  that of \ATLsci we chose here a different, but equivalent, presentation. 
\end{remark}

\begin{remark}
  \label{rem-perf-quant}
  Note that when $n$ is fixed, the propositional quantifier with
  perfect information from \QCTLs
 is equivalent to  the \QCTLsi quantifier
that  observes all the components, \ie, the quantifier parameterised
with observation $[n]$. 
\end{remark}

The model-checking problem for \QCTLsi is the following: given a
closed \QCTLsi formula $\phi$ and a finite pointed \CKS $(\CKS,\state)$, decide
whether $\CKS,\state\modelst\phi$.

We now define the class of \QCTLsi formulas for which the
model-checking problem is known to be decidable with the tree semantics.

\begin{definition}
  \label{def-hier-formulas}
    A \QCTLsi formula $\phi$ is \emph{hierarchical} if for all
  subformulas $\phi_{1},\phi_{2}$ of the form
  $\phi_{1}=\existsp[p_{1}]{\obs_{1}}\phi'_{1}$ and
  $\phi_{2}=\existsp[p_{2}]{\obs_{2}}\phi'_{2}$ where  
  $\phi_{2}$
  is a subformula of $\phi'_{1}$, we have $\obs_{1}\subseteq\obs_{2}$.
\end{definition}

The following result is proved in~\cite{berthon2016qctli}, where \QCTLsih is the set of hierarchical \QCTLsi formulas:

\begin{theorem}[\cite{berthon2016qctli}]
  \label{lab-theo-QCTLi}
  Model checking  \QCTLsih with tree semantics is decidable.
\end{theorem}


\subsection{Model checking \ATLsscititle}
\label{sec-mc-ATLsc}

We establish that model checking \ATLssci is decidable on a class of
instances whose definition 
 relies on the notion of
\emph{hierarchical observation}.

\begin{definition}
  \label{def-hierarchical}
  Let  $\CGS=(\setpos,\trans,\val,\{\obseq\}_{\ag\in\Ag})$ be a \CGS,
  and let $a,b\in\Ag$ be two agents. \emph{Agent $a$ observes no more
    than agent $b$ in $\CGS$}, written $a\obsless b$, if for every
  pair of positions $\pos,\pos'\in\setpos$, $\pos\obseq[b]\pos'$ implies
  $\pos\obseq[a]\pos'$. We say that $A\subseteq \Ag$ is
  \emph{hierarchical in $\CGS$} if
  $\obsless$ is a total preorder on $A$.
\end{definition}

If a set of agents $A$ is hierarchical in a \CGS $\CGS$, we thus may
talk about maximal and minimal agents in $A$, referring to 
maximal and minimal elements of $A$ for the relation $\obsless$.

The essence of the requirement that makes the problem decidable is the
same as for the decidability result on \QCTLsi
(Theorem~\ref{lab-theo-QCTLi}): nesting of quantifiers (here, strategy quantifiers)
should be hierarchical, with those observing more inside those
observing less. However, unlike in \QCTLsi, in \ATLssci observations
are not part of formulas, but rather they are given by the models. We thus
define the notion of hierarchical \ATLssci formula with
respect to a  \CGS:

\begin{definition}
  \label{def-hierarchical-formula}
Let $\Phi$ be  an \ATLssci formula and let $\CGS$ be a \CGS. We say that
$\Phi$ is \emph{hierarchical in $\CGS$} if:
\begin{itemize}
\item for every
  subformula $\phi$ of the form
  $\phi=\Estratsc[A]\phi'$, $A$ is hierarchical in $\CGS$, and
\item for all
  subformulas $\phi_{1},\phi_{2}$ of the form
  $\phi_{1}=\Estratsc[A_{1}]\phi'_{1}$ and
  $\phi_{2}=\Estratsc[A_{2}]\phi'_{2}$ where  
  $\phi_{2}$
  is a subformula of $\phi'_{1}$, maximal agents of $A_{1}$  observe no
  more than minimal agents of $A_{2}$.
\end{itemize}
   An instance $(\Phi,(\CGS,\pos))$ of the model-checking
problem for \ATLssci is \emph{hierarchical} if $\Phi$ is hierarchical
in $\CGS$.
\end{definition}

In the rest of the section we establish
the following:

\begin{theorem}
  \label{theo-ATLsc}
  Model checking \ATLssci is decidable on the class of hierarchical instances.
\end{theorem}

We build upon the proof in
\cite{DBLP:journals/iandc/LaroussinieM15} that
establishes the decidability of the model-checking problem for \ATLssc
by reduction to the model-checking problem for \QCTLs. The main difference is that
we reduce to the model-checking problem for \QCTLsi instead, using
quantifiers parameterised with observations corresponding to agents'
observations. 
We also need to make a couple of adjustments to obtain formulas in the decidable
fragment \QCTLsih.  

Let $(\Phi,(\CGS,\pos_{\init}))$ be a hierarchical instance of the \ATLssci
model-checking problem, where
$\CGS=(\setpos,\trans,\val,\{\obseq\}_{\ag\in\Ag})$ is a \CGS over
$\APf$.  In the reduction we will  transform  $\Phi$ into an
equivalent \QCTLsi formula $\Phi'$ in which we need to
refer to the current position in the model  $\CGS$, and also to talk about moves taken
by agents. To do so, we consider the additional sets of atomic
propositions $\APv\egdef\{p_{\pos}\mid\pos\in\setpos\}$ and
$\APm\egdef\{p^{\ag}_{\mov}\mid\ag\in\Ag \mbox{ and }\mov\in\setmoves\}$,
that we take disjoint from $\APf$.

First we  
define the \CKS $\CKS_{\CGS}$ on which  $\Phi'$
will be evaluated. Since the models of the two logics use different
ways to represent
imperfect information (equivalence relations on positions for \CGS and
local states for \CKS) this requires a bit of work. First, for each $\pos\in\setpos$ and $\ag\in\Ag$,  let us
define $\eqc{\ag}$ as the equivalence class of $\pos$ for relation
$\obseq$. Now,  noting
$\Ag=\{\ag_{1},\ldots,\ag_{n}\}$, we define for each $i\in [n]$ the set
$\setlstates_{i}\egdef\{\eqc{\ag_{i}}\mid\pos\in\setpos\}$ of local
states for agent~$\ag_{i}$.  Since
we need to know the actual position of the \CGS to define the dynamics,  we
also let $\setlstates_{n+1}\egdef\setpos$. States of
$\CKS_{\CGS}$ will thus be tuples in $\setlstates_{1}\times\ldots\times\setlstates_{n}\times\setlstates_{n+1}$.
For each $\pos\in\CGS$, let
$\state_{\pos}\egdef(\eqc{\ag_{1}},\ldots,\eqc{\ag_{n}},\pos)$ be its
corresponding state in $\CKS_{\CGS}$.

We can now define $\CKS_{\CGS}\egdef(\setstates,\relation,\lab')$, where
\begin{itemize}
\item $\setstates\egdef\{\state_{\pos} \mid \pos\in\setpos\}$,
\item $\relation\egdef\{(\state_{\pos},\state_{\pos'})\mid
  \exists\jmov\in\Mov^{\Ag} \mbox{ s.t. }\trans(\pos,\jmov)=\pos'\}$, and
\item $\lab'(\state_{\pos})\egdef\val(\pos)\union \{p_{\pos}\}$.
\end{itemize}

To make the connection between finite plays in $\CGS$ and nodes in
tree unfoldings of $\CKS_{\CGS}$, let us define, for every finite play
$\fplay=\pos_{0}\ldots\pos_{k}$, the node $\noeud_{\fplay}\egdef
\state_{\pos_{0}}\ldots \state_{\pos_{k}}$ in
$\unfold[\CKS_{\CGS}]{\state_{\pos_{0}}}$ (which exists, by definition of
$\CKS_{\CGS}$ and of tree unfoldings). Observe that the mapping
$\fplay\mapsto\noeud_{\fplay}$ is in fact a bijection between the set
of finite plays starting in a given position $\pos$ and the set of
nodes in $\unfold[\CKS_{\CGS}]{\state_{\pos}}$.

Now it should be clear that giving to a propositional quantifier in
\QCTLsi observation $\obs_{i}\egdef \{i\}$, for $i\in[n]$, amounts to
giving him the same observation as agent~$\ag_{i}$. Formally, one can
 prove the following lemma, simply by applying the definitions
of observational equivalence in the two frameworks:
\begin{lemma}
  \label{lem-obs}
  For all finite plays $\fplay,\fplay'$ starting in 
  $\pos$, $\fplay\obseq[\ag_{i}]\fplay'$ iff
  $\noeud_{\fplay}\oequivt[\obs_{i}]\noeud_{\fplay'}$ in $\unfold[\CKS_{\CGS}]{\state_{\pos}}$.
\end{lemma}

We now describe the translation\footnote{Here we abuse language: the
  construction depends on the model $\CGS$ and is therefore not a
  translation in the usual sense.} from \ATLsci formulas to \QCTLsi
formulas. First we recall the translation from~\cite{DBLP:journals/iandc/LaroussinieM15} for the perfect-information
case. 

The translation  from \ATLsc to \QCTLs is parameterised by a coalition
$\coala\subset\Ag$, that conveys the set of agents who are currently
bound to a strategy. It is defined by induction on $\Phi$ as follows:
\begin{align*}
  \tr{p}&\egdef p & \tr{\neg \phi}&\egdef \neg \tr{\phi}\\
  \tr{\phi\ou\phi'} &\egdef \tr{\phi}\ou\tr{\phi'} &
  \tr{\relstrat\phi}&\egdef \tr[\coala\setminus\coal]{\phi} \\
  \tr{\X\phi}&\egdef \X\tr{\phi} & \tr{\phi\until\phi'}&\egdef \tr{\phi}\until\tr{\phi'}
\end{align*}
The only non-trivial case is for formulas of the form
$\Estratsc\phi$. For the rest of the section, we let 
$\Mov=\{\mov_{1},\ldots,\mov_{\maxmov}\}$. Now, if $\coal=\{\ag_{i_{1}},\ldots,\ag_{i_{k}}\}$,
we define
\begin{align*}
\tr{\Estratsc\phi}\egdef & \;\exists \mov_{1}^{\ag_{i_{1}}} \ldots
\mov_{\maxmov}^{\ag_{i_{1}}}\ldots \mov_{1}^{\ag_{i_{k}}} \ldots
\mov_{\maxmov}^{\ag_{i_{k}}} \pout.\\
& \left ( \phistrat{\coal} \et \phiout{\coal\union\coala} \et \A (\G
  \pout \impl \tr[\coal\union\coala]{\phi})  \right ),
\end{align*}
where
\[
\phistrat{\coal}\egdef\biget_{\ag\in\coal}\A\G\bigou_{\mov\in\Mov}(\mov^{\ag}\et\biget_{\mov'\neq\mov}\neg\mov'^{\ag})
\]
and
\begin{align*}
 \phiout{\coal} \egdef & \; \pout \et \A\G \left[\neg \pout \impl \A\X\neg
  \pout \right] \et \\ & \A\G\left[\pout \impl \vphantom{\bigou_{\pos\in\setpos}   
   \bigou_{\jmov\in\Mov^{\coal}} \left(p_{\pos} \et p_{\jmov} \et \A\X \left(
       \bigou_{\pos'\in AA}p_{\pos'} \leftrightarrow \pout \right) \right) } \right.   \left. \bigou_{\pos\in\setpos}   
   \bigou_{\jmov\in\Mov^{\coal}} \left(p_{\pos} \et p_{\jmov} \et \A\X \left(
       \bigou_{\pos'\in \trans(\pos,\jmov)}p_{\pos'} \leftrightarrow \pout \right) \right) \right].
\end{align*}

In $\phiout{\coal}$, for $\jmov=(\mov_{\ag})_{\ag\in\coal}\in\Mov^{\coal}$,
notation $p_{\jmov}$ stands for
the propositional formula $\biget_{\ag\in\coal}\mov_{\ag}^{\ag}$ which characterizes the joint
move $\jmov$ that agents in $\coal$ play in $\pos$.
Also, $\trans(\pos,\jmov)$ is
the set of possible next positions when the current one is $\pos$ and
agents in $\coal$ play $\jmov$, and it is defined as
$\trans(\pos,\jmov)\egdef\{\trans(\pos,(\jmov,\jmov'))\mid
\jmov'\in\Mov^{\Ag\setminus\coal}\}$.

The idea of this translation is the following: first, for each agent
 $\ag\in\coal$ and each possible move $\mov\in\Mov$, an existential
 quantification on the  atomic proposition $\mov^{\ag}$ ``chooses''
 for each finite play $\fplay$ of $(\CGS,\pos_{\init})$ (or, equivalently,
 for each node $\noeud_{\fplay}$ of the $\unfold[\CKS_{\CGS}]{\state_{\pos_{\init}}}$)
 whether agent $\ag$
 plays  move $\mov$ in $\fplay$ or not. Formula $\phistrat{\coal}$ ensures that each
 agent $\ag$
 chooses exactly one move in each finite play, and thus that
 atomic propositions $\mov^{\ag}$ characterise a strategy for her.
 An  atomic proposition $\pout$ is then used
 to mark the paths that follow the 
 currently fixed strategies: formula $\phiout{\coal\union\coala}$
 states that $\pout$ marks exactly the outcome of
 strategies just chosen for agents
 in $\coal$, as well as those of agents in $\coala$, that were
 chosen previously by a strategy quantifier ``higher'' in $\Phi$. 
 
 Note that we simplified slightly $\phistrat{\coal}$ and
$\phiout{\coal}$, using the fact that unlike
in~\cite{DBLP:journals/iandc/LaroussinieM15}, we have assumed in our
definition of \CGSs that the set of available
moves is the same for all agents in all positions (see
Footnote~\ref{footnote-models}). 

It is proven in~\cite{DBLP:journals/iandc/LaroussinieM15} that this
translation is correct, in the sense that for every \ATLsc
closed formula $\phi$ and pointed
perfect-information concurrent game structure $(\CGS,\pos)$,
letting $\KS_{\CGS}$ be as described above but removing the local
states for all agents and keeping only the $\setlstates_{n+1}$ component,  we have:
\[\CGS,\pos\models\phi \mbox{ iff }
\unfold[\CKS_{\CGS}]{\state_{\pos}}\modelst \tr[\emptyset]{\phi}.\]
We now explain how to adapt this translation to the case of imperfect information.
Observe that the only difference between \ATLssc and \ATLssci is that in
the latter, strategies  must be defined uniformly
over indistinguishable finite plays, \ie,  a strategy $\strat$ for an agent
$\ag$ must be such that if $\fplay\obseq\fplay'$, then $\strat(\fplay)=\strat(\fplay')$.
To enforce that the strategies coded by atomic propositions
$\mov^{\ag}$ in $\tr{\Estratsc \phi}$ are uniform, we use the
propositional quantifiers with partial observation of
\QCTLsi. Formally, we define a translation $\tri{~~~~}$ from \ATLssci to
\QCTLsi. It is defined exactly as the one from \ATLssc to \QCTLs,
except for the following inductive case.

If $\coal=\{\ag_{i_{1}},\ldots,\ag_{i_{k}}\}$ we let
\begin{align*}
\tri{\Estratsc\phi}\egdef & \;\exists^{\obs_{i_{1}}} \mov_{1}^{\ag_{i_{1}}} \ldots
\mov_{\maxmov}^{\ag_{i_{1}}}\ldots \exists^{\obs_{i_{k}}}\mov_{1}^{\ag_{i_{k}}} \ldots
\mov_{\maxmov}^{\ag_{i_{k}}} \exists\pout.\\
& \left ( \phistrat{\coal} \et \phiout{\coal\union\coala} \et \A (\G
  \pout \impl \tri[\coal\union\coala]{\phi})  \right ),
\end{align*}
where $\phistrat{\coal}$ and $\phiout{\coal}$ are defined as before,
and $\exists \pout$ is a macro for $\exists^{\{1,\ldots,n+1\}}\pout$
(see Remark~\ref{rem-perf-quant}).

So the only difference from the previous translation is that now, the
labelling of each atomic proposition $\mov^{\ag_{i}}$ must be
$\obs_{i}$-uniform. 
This means that if two nodes $\noeud$ and
$\noeud'$ in $\unfold[\CKS_{\CGS}]{\state_{\pos_{\init}}}$ are
$\obs_{i}$-indistinguishable, then $\noeud$ is labelled with
$\mov^{\ag_{i}}$ if and only if $\noeud'$ also is. In other words,
 in the strategy coded by atomic propositions $\mov^{\ag_{i}}$, agent $\ag_{i}$ plays $\mov$ in $\noeud$ if and only
if she also plays it in $\noeud'$, and thus this strategy 
 is uniform (recall that, by Lemma~\ref{lem-obs}, observation $\obs_{i}$
correctly reflects agent $\ag_{i}$'s observation in
$\unfold[\CKS_{\CGS}]{\state_{\pos_{\init}}}$). It is then clear
 that this translation is correct:
 \begin{equation}
   \label{eq-1}
\CGS,\pos_{\init}\models\Phi \mbox{ iff
}\unfold[\CKS_{\CGS}]{\state_{\pos_{\init}}}\modelst \tri[\emptyset]{\Phi}.   
 \end{equation}

However, even if we have taken $(\Phi,(\CGS,\pos_{\init}))$ to be a
hierarchical instance, $\tri[\emptyset]{\Phi}$ is not in the decidable
fragment \QCTLsih.
Indeed, with the current definition
of observations $\{\obs_{i}\}_{i\in[n]}$,
hierarchical observation in $\CGS$ does
not imply hierarchical observation in $\CKS_{\CGS}$: since $\obs_{i}=\{i\}$, for
$i\neq j$ it is never the case that $\obs_{i}\subseteq
\obs_{j}$. Still, we note that if agent $\ag_{j}$ observes no more
than agent $\ag_{i}$, then letting $\ag_{i}$ see also what agent
$\ag_{j}$ sees does not increase her knowledge of the
situation:

\begin{lemma}
  \label{lem-more}
  If $\ag_{j}\obsless\ag_{i}$, then for all finite plays
  $\fplay,\fplay'$ that start in the same position,
  $\noeud_{\fplay}\oequivt[\obs_{i}]\noeud_{\fplay'}$ iff   $\noeud_{\fplay}\oequivt[\obs_{i}\union\obs_{j}]\noeud_{\fplay'}$.
\end{lemma}

\begin{proof}
  Assume that $\ag_{j}\obsless\ag_{i}$.  It is enough to see that for
  every pair of states $\state_{\pos},\state_{\pos'}$ in
  $\CKS_{\CGS}$, we have $\state_{\pos}\oequiv[\obs_{i}]\state_{\pos'}$
  iff
  $\state_{\pos}\oequiv[\obs_{i}\union\obs_{j}]\state_{\pos'}$. The
  right-to-left implication is obvious: if two states have the same
  $i$-th and $j$-th components, in particular they have the same
  $i$-th component. For the other direction, assume that
  $\state_{\pos}\oequiv[\obs_{i}]\state_{\pos'}$. This means that
  $\eqc{\ag_{i}}=\eqc[\pos']{\ag_{i}}$, and thus that
  $\pos\obseq[\ag_{i}]\pos'$. Since $\ag_{j}\obsless\ag_{i}$, we also
  have that   $\pos\obseq[\ag_{j}]\pos'$, and thus that
  $\eqc{\ag_{j}}=\eqc[\pos']{\ag_{j}}$, and it follows that   $\state_{\pos}\oequiv[\obs_{i}\union\obs_{j}]\state_{\pos'}$.
\end{proof}

In the light of this Lemma~\ref{lem-more}, we can safely redefine
observations as follows: for each $i\in [n]$, we let
\[\obs'_{i}\egdef\bigunion_{j\mid \ag_{j}\obsless\ag_{i}}\obs_{j}.\] 
Observe that in fact $\obs'_{i}=\{j\mid \ag_{j}\obsless\ag_{i}\}$.
Informally, a quantifier with observation $\obs'_{i}$ sees what agent
$\ag_{i}$ observes (note that $\obsless$ is reflexive), as well as
what agents that see no more than
$\ag_{i}$ observe.

Let us define a new version of the translation $\tri{~~~~}$.
First, $\Phi$ being hierarchical in $\CGS$, for each subformula of $\Phi$ of
the form $\Estratsc\phi$ we have that  $\coal$ is hierarchical in $\CGS$.
It is thus possible to choose for agents in $\coal$ an indexing
$\coal=\{\ag_{i_{1}},\ldots,\ag_{i_{k}}\}$ 
such that for all $1\leq c < d\leq k$, we have
$\ag_{i_{c}}\obsless\ag_{i_{d}}$.

Now the translation remains the same as before except for the following inductive case:

If $\coal=\{\ag_{i_{1}},\ldots,\ag_{i_{k}}\}$,
where  for all $1\leq c < d\leq k$, we have
$\ag_{i_{c}}\obsless\ag_{i_{d}}$,
we let
\begin{align*}
\tri{\Estratsc\phi}\egdef &\; \exists^{\obs'_{i_{1}}} \mov_{1}^{\ag_{i_{1}}} \ldots
\mov_{\maxmov}^{\ag_{i_{1}}}\ldots \exists^{\obs'_{i_{k}}}\mov_{1}^{\ag_{i_{k}}} \ldots
\mov_{\maxmov}^{\ag_{i_{k}}} \exists\pout.\\
& \left ( \phistrat{\coal} \et \phiout{\coal\union\coala} \et \A (\G
  \pout \impl \tri[\coal\union\coala]{\phi})  \right ),
\end{align*}
where $\phistrat{\coal}$ and $\phiout{\coal}$ are defined as before.

From Lemma~\ref{lem-more} we have that this new translation is still
correct in the sense of Equation~(\ref{eq-1}). In addition, for all
$1\leq c < d \leq k$ we have
$\obs'_{i_{c}}\subseteq \obs'_{i_{d}}$.

Now consider formula $\tri[\emptyset]{\Phi}$. Because
$\Phi$ is hierarchical in $\CGS$, 
for every pair of
  subformulas $\phi_{1},\phi_{2}$ of the form
  $\phi_{1}=\Estratsc[A_{1}]\phi'_{1}$ and
  $\phi_{2}=\Estratsc[A_{2}]\phi'_{2}$ where  
  $\phi_{2}$
  is a subformula of $\phi'_{1}$,  maximal agents of $A_{1}$  observe no
  more than minimal agents of $A_{2}$.
  It is then easy to see that $\tri[\emptyset]{\Phi}$ would be
  hierarchical if there were not the perfect-information
  quantifications on atomic proposition $\pout$ that break the
  monotony of observations along subformulas when there are nested strategic quantifiers.
  We explain how to remedy this last problem.

We remove altogether  proposition $\pout$, and we use
instead the formula $\psiout{\coal}$ defined below  to characterise
which paths are in the outcome of the currently-fixed strategies:
\[\psiout{\coal}\egdef \G\left(
  \biget_{\pos\in\setpos}\biget_{\jmov\in\Mov^{\coal}} p_{\pos} \et
  p_{\jmov}\impl \X \bigou_{\pos'\in\trans(\pos,\jmov)} p_{\pos'}
\right).\]

Clearly, this formula holds in a path $\tpath$ of
$\unfold[\CKS_{\CGS}]{\state_{\pos_{\init}}}$ marked with propositions
$\mov^{\ag}$ characterising strategies for agents in $\coal$, if at
each point along $\tpath$ corresponding to some position $\pos$, the
next point in $\tpath$ corresponds to a position $\pos'$ that can be attained
from $\pos$ when agents in $\coal$ each play the move prescribed by
their current strategy.
The last modification to $\tri{~~~~}$ is thus the following:

If $\coal=\{\ag_{i_{1}},\ldots,\ag_{i_{k}}\}$,
where  for all $1\leq c < d\leq k$, we have $\ag_{i_{c}}\obsless\ag_{i_{d}}$,
we let
\begin{align*}
\tri{\Estratsc\phi}\egdef & \;\exists^{\obs'_{i_{1}}} \mov_{1}^{\ag_{i_{1}}} \ldots
\mov_{\maxmov}^{\ag_{i_{1}}}\ldots \exists^{\obs'_{i_{k}}}\mov_{1}^{\ag_{i_{k}}} \ldots
\mov_{\maxmov}^{\ag_{i_{k}}}.  \phistrat{\coal} \et \A \left(\psiout{\coal\union\coala} \impl \tri[\coal\union\coala]{\phi}\right),
\end{align*}
where $\phistrat{\coal}$ is defined as before.

It follows from the above considerations that this translation is
still correct in the sense of Equation~(\ref{eq-1}), and one can check
that $\tri[\emptyset]{\Phi}$ is a hierarchical \QCTLsi formula.
We conclude the proof by recalling that by
Theorem~\ref{lab-theo-QCTLi}, model checking \QCTLsih is decidable.

Concerning complexity, model checking \ATLsc being already
nonelementary~\cite{DBLP:journals/iandc/LaroussinieM15}, so is it for
\ATLsci. 
  
 


\section{Conclusion}
\label{sec-conclusion}

In this work we established new decidability results for the
model-checking problem of \ATLs with imperfect information and perfect
recall as well as its extension with strategy context. Should new
decidable classes of multiplayer games with imperfect information  be
discovered, and assuming the reasonable property of closure under
initial shifting, our transfer theorem
(Theorem~\ref{theo-meta}) would entail new 
decidability results also for \ATLsi. As for \ATLssci, it would be
interesting to investigate whether a meaningful notion of hierarchical instances
based on, \eg, dynamic or recurring hierarchical information instead of
hierarchical observation as  here, could  lead to
stronger decidability results.  
%



\end{document}